\documentclass[a4paper,USenglish]{lipics}

\usepackage{microtype}

\usepackage{xspace}

\usepackage{enumitem}
\usepackage{appendix}

\theoremstyle{plain}
\newtheorem{observation}[theorem]{Observation}
\newtheorem{claim}[theorem]{Claim}
\newtheorem*{relemma}{Lemma}
\newtheorem*{retheorem}{Theorem}

\title{Sparsification Upper and Lower Bounds for Graphs Problems and Not-All-Equal SAT\footnote{This work was supported by NWO Veni grant ``Frontiers in Parameterized Preprocessing'' and NWO Gravity grant ``Networks''.}}
\titlerunning{Sparsification Upper and Lower Bounds for Graphs Problems and Not-All-Equal SAT} 

\author{Bart M.P. Jansen}
\author{Astrid Pieterse}
\affil{Eindhoven University of Technology\\
  P.\,O.\, Box 513, Eindhoven, The Netherlands\\
  \texttt{b.m.p.jansen@tue.nl}, \texttt{a.pieterse@tue.nl}}
\authorrunning{B.\,M.\,P.\, Jansen and A. Pieterse} 

\Copyright{Bart M. P. Jansen and Astrid Pieterse}

\subjclass{F.2.2 Nonnumerical Algorithms and Problems, G.2.2 Graph Theory}
\keywords{sparsification, graph coloring, Hamiltonian cycle, satisfiability}

\serieslogo{}
\volumeinfo
  {}
  {2}
  {}
  {1}
  {1}
  {1}
\EventShortName{}

\newcommand{\todo}[1][]{%
  \ifx/#1/%
    \textcolor{red}{TODO!}%
  \else%
    \textcolor{red}{todo: #1}%
  \fi%
}

\newcommand{\vect}[1]{\mathbf{#1}}
\let\oldcite=\cite

\renewcommand{\cite}[2][TODO]{%
    \oldcite{#2}%
}
\newcommand{\ID}{\ensuremath{\mathrm{\textsc{{id}}}}}
\newcommand*{\etal}{\textit{et al.\@}}

\newcommand{\problem}[1]{\textsc{#1}}

\newcommand{\defproblem}[3]{
 \vspace{1mm}
\noindent\fbox{
 \begin{minipage}{0.96\textwidth}
 \begin{tabular*}{\textwidth}{@{\extracolsep{\fill}}lr} #1 &  \\ \end{tabular*}
 {\bf{Input:}} #2 \\
 {\bf{Question:}} #3
 \end{minipage}
 }
 \vspace{1mm}
}

\newcommand{\q}{t'}
\newcommand{\FourColoring}{\textsc{$4$-Coloring}\xspace}

\newenvironment{mydef}{\begin{definition}}{\end{definition}}
\newenvironment{mythm}{\begin{theorem}}{\end{theorem}}
\newenvironment{mylem}{\begin{lemma}}{\end{lemma}}

\begin{document}

\maketitle
\newcommand{\Oh}{\mathcal{O}}
\newcommand{\containment}{\ensuremath{\mathsf{NP \subseteq coNP/poly}}\xspace}
\newcommand{\ncontainment}{\ensuremath{\mathsf{NP \nsubseteq coNP/poly}}\xspace}
\newcommand{\yes}{\emph{yes}\xspace}
\newcommand{\no}{\emph{no}\xspace}
\newcommand{\true}{\emph{true}\xspace}
\newcommand{\false}{\emph{false}\xspace}

\begin{abstract}
We present several sparsification lower and upper bounds for classic problems in graph theory and logic. For the problems \textsc{$4$-Coloring}, \textsc{(Directed) Hamiltonian Cycle}, and \textsc{(Connected) Dominating Set}, we prove that there is no polynomial-time algorithm that reduces any $n$-vertex input to an equivalent instance, of an arbitrary problem, with bitsize~$\Oh(n^{2-\varepsilon})$ for~$\varepsilon > 0$, unless \containment and the polynomial-time hierarchy collapses. These results imply that existing linear-vertex kernels for \textsc{$k$-Nonblocker} and \textsc{$k$-Max Leaf Spanning Tree} (the parametric duals of \textsc{(Connected) Dominating Set}) cannot be improved to have~$\Oh(k^{2-\varepsilon})$ edges, unless \containment. We also present a positive result and exhibit a non-trivial sparsification algorithm for \textsc{$d$-Not-All-Equal-SAT}. We give an algorithm that reduces an $n$-variable input with clauses of size at most~$d$ to an equivalent input with~$\Oh(n^{d-1})$ clauses, for any fixed~$d$. Our algorithm is based on a linear-algebraic proof of Lov\'{a}sz that bounds the number of hyperedges in critically $3$-chromatic $d$-uniform $n$-vertex hypergraphs by~$\binom{n}{d-1}$. We show that our kernel is tight under the assumption that \ncontainment.
\end{abstract}

\section{Introduction}
\label{sec:introduction}
\subparagraph*{Background.}
Sparsification refers to the method of reducing an object such as a graph or CNF-formula to an equivalent object that is less dense, that is, an object in which the ratio of edges to vertices (or clauses to variables) is smaller. The notion is fruitful in theoretical~\cite{ImpagliazzoPZ01} and practical (cf.~\cite{EppsteinGIN97}) settings when working with (hyper)graphs and formulas. The theory of kernelization, originating from the field of parameterized complexity theory, can be used to analyze the limits of polynomial-time sparsification. Using tools developed in the last five years, it has become possible to address questions such as: ``Is there a polynomial-time algorithm that reduces an $n$-vertex instance of my favorite graph problem to an equivalent instance with a subquadratic number of edges?''

The impetus for this line of analysis was given by an influential paper by Dell and van Melkebeek~\cite{DellM14} (conference version in 2010). One of their main results states that if there is an~$\varepsilon > 0$ and a polynomial-time algorithm that reduces any $n$-vertex instance of \textsc{Vertex Cover} to an equivalent instance, of an arbitrary problem, that can be encoded in~$\Oh(n^{2-\varepsilon})$ bits, then \containment and the polynomial-time hierarchy collapses. Since any nontrivial input~$(G,k)$ of \textsc{Vertex Cover} has~$k \leq n = |V(G)|$, their result implies that the number of edges in the $2k$-vertex kernel for \textsc{$k$-Vertex Cover}~\cite{NemhauserT75} cannot be improved to~$\Oh(k^{2-\varepsilon})$ unless \containment.

Using related techniques, Dell and van Melkebeek also proved important lower bounds for \textsc{$d$-cnf-sat} problems: testing the satisfiability of a propositional formula in CNF form, where each clause has at most~$d$ literals. They proved that for every fixed integer~$d \geq 3$, the existence of a polynomial-time algorithm that reduces any $n$-variable instance of \textsc{$d$-cnf-sat} to an equivalent instance, of an arbitrary problem, with~$\Oh(n^{d-\varepsilon})$ bits, for some~$\varepsilon > 0$ implies \containment. Their lower bound is tight: there are~$\Oh(n^d)$ possible clauses of size~$d$ over~$n$ variables, allowing an instance to be represented by a vector of~$\Oh(n^d)$ bits that specifies for each clause whether or not it is present.

\subparagraph*{Our results.}
We continue this line of investigation and analyze sparsification for several classic problems in graph theory and logic. We obtain several sparsification lower bounds that imply that the quadratic number of edges in existing linear-vertex kernels is likely to be unavoidable. When it comes to problems from logic, we give the---to the best of our knowledge---first example of a problem that \emph{does} admit nontrivial sparsification: \textsc{$d$-Not-All-Equal-SAT}. We also provide a matching lower bound.

The first problem we consider is \textsc{$4$-Coloring}, which asks whether the input graph has a proper vertex coloring with~$4$ colors. Using several new gadgets, we give a cross-composition~\cite{BodlaenderJK14} to show that the problem has no compression of size~$\Oh(n^{2-\varepsilon})$ unless \containment. To obtain the lower bound, we give a polynomial-time construction that embeds the logical \textsc{or} of a series of~$t$ size-$n$ inputs of an NP-hard problem into a graph~$G'$ with~$\Oh(\sqrt{t} \cdot n^{\Oh(1)})$ vertices, such that~$G'$ has a proper $4$-coloring if and only if there is a \yes-instance among the inputs. The main structure of the reduction follows the approach of Dell and Marx~\cite{DellM12}: we create a table with two rows and~$\Oh(\sqrt{t})$ columns and~$\Oh(n^{\Oh(1)})$ vertices in each cell. For each way of picking one cell from each row, we aim to embed one instance into the edge set between the corresponding groups of vertices. When the NP-hard starting problem is chosen such that the~$t$ inputs each decompose into two induced subgraphs with a simple structure, one can create the vertex groups and their connections such that for each pair of cells~$(i,j)$, the subgraph they induce represents the~$i \cdot \sqrt{t} + j$-th input. If there is a \yes-instance among the inputs, this leads to a pair of cells that can be properly colored in a structured way. The challenging part of the reduction is to ensure that the edges in the graph corresponding to \no-inputs do not give conflicts when extending this partial coloring to the entire graph.

The next problem we attack is \textsc{Hamiltonian Cycle}. We rule out compressions of size~$\Oh(n^{2-\varepsilon})$ for the directed and undirected variant of the problem, assuming \ncontainment. The construction is inspired by kernelization lower bounds for \textsc{Directed Hamiltonian Cycle} parameterized by the vertex-deletion distance to a directed graph whose underlying undirected graph is a path~\cite{BodlaenderJK13b}.

By combining gadgets from kernelization lower bounds for two different parameterizations of \textsc{Red Blue Dominating Set}, we prove that there is no compression of size~$\Oh(n^{2-\varepsilon})$ for \textsc{Dominating Set} unless \containment. The same construction rules out subquadratic compressions for \textsc{Connected Dominating Set}. These lower bounds have implications for the kernelization complexity of the parametric duals \textsc{Nonblocker} and \textsc{Max Leaf Spanning Tree} of \textsc{(Connected) Dominating Set}. For both \textsc{Nonblocker} and \textsc{Max Leaf} there are kernels with~$\Oh(k)$ vertices~\cite{Wiedermann2006nonblocker,Estivill-CastroFLR05} that have~$\Theta(k^2)$ edges. Our lower bounds imply that the number of edges in these kernels cannot be improved to~$\Oh(k^{2-\varepsilon})$, unless \containment.

The final family of problems we consider is \textsc{$d$-Not-All-Equal-SAT} for fixed~$d \geq 4$. The input consists of a formula in CNF form with at most~$d$ literals per clause. The question is whether there is an assignment to the variables such that each clause contains both a variable that evaluates to \true and one that evaluates to \false. There is a simple linear-parameter transformation from \textsc{$d$-cnf-sat} to \textsc{$(d+1)$-nae-sat} that consists of adding one variable that occurs as a positive literal in all clauses. By the results of Dell and van Melkebeek discussed above, this implies that \textsc{$d$-nae-sat} does not admit compressions of size~$\Oh(n^{d-1-\varepsilon})$ unless \containment. We prove the surprising result that this lower bound is tight! A linear-algebraic result due to Lov\'asz~\cite{lovasz}, concerning the size of critically $3$-chromatic $d$-uniform hypergraphs, can be used to give a kernel for \textsc{$d$-nae-sat} with~$\Oh(n^{d-1})$ clauses for every fixed~$d$. The kernel is obtained by computing the basis of an associated matrix and removing the clauses that can be expressed as a linear combination of the basis clauses.

\subparagraph*{Related work.} Dell and Marx introduced the table structure for compression lower bounds~\cite{DellM12} in their study of compression for packing problems. Hermelin and Wu~\cite{HermelinW12} analyzed similar problems. 
Other papers about polynomial kernelization and sparsification lower bounds include~\cite{CyganGH13} and~\cite{Jansen15}. 

\section{Preliminaries}

\newcommand{\N}{\ensuremath{\mathbb{N}}\xspace}
\newcommand{\Q}{\ensuremath{\mathcal{Q}}\xspace}
\newcommand{\eqvr}[0]{\ensuremath{\mathcal{R}}\xspace}
\newcommand{\OR}[0]{\ensuremath{\mathop{\mathrm{\textsc{or}}}}\xspace}

A parameterized problem~$\Q$ is a subset of~$\Sigma^* \times \mathbb{N}$, where~$\Sigma$ is a finite alphabet. Let~$\Q,\Q'\subseteq\Sigma^*\times\N$ be parameterized problems and let~$h\colon\N\to\N$ be a computable function. A \emph{generalized kernel for~$\Q$ into~$\Q'$ of size~$h(k)$} is an algorithm that, on input~$(x,k)\in\Sigma^*\times\N$, takes time polynomial in~$|x|+k$ and outputs an instance~$(x',k')$ such that: 
\begin{enumerate}
\item $|x'|$ and~$k'$ are bounded by~$h(k)$, and 
\item $(x',k') \in \Q'$ if and only if~$(x,k) \in \Q$. 
\end{enumerate}
The algorithm is a \emph{kernel} for~$\Q$ if~$\Q'=\Q$. It is a \emph{polynomial (generalized) kernel} if~$h(k)$ is a polynomial.

Since a polynomial-time reduction to an equivalent sparse instance yields a generalized kernel, we will use the concept of generalized kernels in the remainder of this paper to prove the non-existence of such sparsification algorithms. We employ the cross-composition framework by Bodlaender \emph{et al.}~\cite{BodlaenderJK14}, which builds on earlier work by several authors~\cite{BodlaenderDFH09,DellM14,FortnowS11}.

\begin{definition}[Polynomial equivalence relation] \label{definition:eqvr}
An equivalence relation~\eqvr on~$\Sigma^*$ is called a \emph{polynomial equivalence relation} if the following conditions hold.
 \begin{enumerate}
 \item There is an algorithm that, given two strings~$x,y \in \Sigma^*$, decides whether~$x$ and~$y$ belong to the same equivalence class in time polynomial in~$|x| + |y|$.
 \item For any finite set~$S \subseteq \Sigma^*$ the equivalence relation~$\eqvr$ partitions the elements of~$S$ into a number of classes that is polynomially bounded in the size of the largest element of~$S$.
 \end{enumerate}
\end{definition}

\begin{definition}[Cross-composition]\label{definition:crosscomposition}
Let~$L\subseteq\Sigma^*$ be a language, let~$\eqvr$ be a polynomial equivalence relation on~$\Sigma^*$, let~$\Q\subseteq\Sigma^*\times\N$ be a parameterized problem, and let~$f \colon \N \to \N$ be a function. An \emph{\OR-cross-com\-position of~$L$ into~$\Q$} (with respect to \eqvr) \emph{of cost~$f(t)$} is an algorithm that, given~$t$ instances~$x_1, x_2, \ldots, x_t \in \Sigma^*$ of~$L$ belonging to the same equivalence class of~$\eqvr$, takes time polynomial in~$\sum _{i=1}^t |x_i|$ and outputs an instance~$(y,k) \in \Sigma^* \times \mathbb{N}$ such that: 
\begin{enumerate}
 \item the parameter~$k$ is bounded by $\Oh(f(t)\cdot(\max_i|x_i|)^c)$, where~$c$ is some constant independent of~$t$, and 
 \item $(y,k) \in \Q$ if and only if there is an~$i \in [t]$ such that~$x_i \in L$.\label{property:OR}
 \end{enumerate}
\end{definition}

\begin{theorem}[{\cite[Theorem 6]{BodlaenderJK14}}] \label{thm:cross_composition_LB}
Let~$L\subseteq\Sigma^*$ be a language, let~$\Q\subseteq\Sigma^*\times\N$ be a parameterized problem, and let~$d,\varepsilon$ be positive reals. If~$L$ is NP-hard under Karp reductions, has an \OR-cross-composition into~$\Q$ with cost~$f(t)=t^{1/d+o(1)}$, where~$t$ denotes the number of instances, and~$\Q$ has a polynomial (generalized) kernelization with size bound~$\Oh(k^{d-\varepsilon})$, then \containment.
\end{theorem}

For~$r \in \N$ we will refer to an \OR-cross-composition of cost~$f(t) = t^{1/r} \log (t)$ as a \emph{degree-$r$ cross-composition}. By Theorem~\ref{definition:crosscomposition}, a degree-$r$ cross-composition can be used to rule out generalized kernels of size~$\Oh(k^{r - \varepsilon})$. We frequently use the fact that a polynomial-time linear-parameter transformation from problem~$\Q$ to~$\Q'$ implies that any generalized kernelization lower bound for~$\Q$, also holds for~$\Q'$ (cf.~\cite{BodlaenderJK14,BodlaenderTY11}). Let $[r]$ be defined as $[r]~:=~\{x~\in~\mathbb{N}~\mid~1~\leq~x~\leq~r\}$. 

\section{\texorpdfstring{$4$-Coloring}{4-Coloring}}
\label{sec:4-coloring}
In this section we analyze the \FourColoring problem, which asks whether it is possible to assign each vertex of the input graph one out of 4 possible colors, such that there is no edge whose endpoints share the same color. We show that \FourColoring does not have a generalized kernel of size $\Oh(n^{2-\varepsilon})$, by giving a degree-$2$ cross-composition from a tailor-made problem that will be introduced below. Before giving the construction, we first present and analyze some of the gadgets that will be needed.

\begin{figure}[t]
\centering
\subfloat[Sub figure 1 list of figures text][Treegadget with no red leaf.]{
\includegraphics{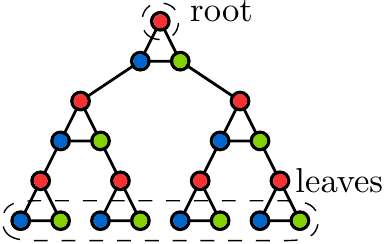}
\label{fig:tree_gadget_lemma_1}}
\quad
\subfloat[Sub figure 2 list of figures text][Treegadget where one of the leaves is red.]{
\includegraphics{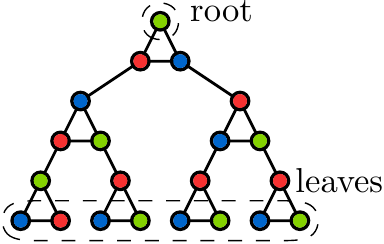}
\label{fig:tree_gadget_lemma_2}}
\quad
\subfloat[Sub figure 2 list of figures text][Triangular gadget.]{
\includegraphics{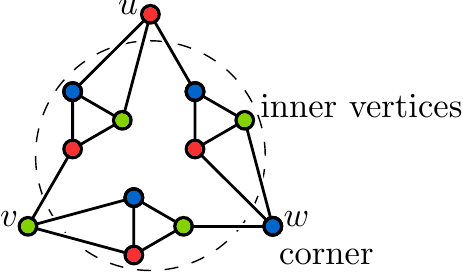}
\label{fig:triangular_gadget}
}
\caption{Used gadgets with example colorings.}
\label{fig:tree_gadget_lemmas}
\end{figure}

\begin{mydef}
A \emph{treegadget} is the graph obtained from a complete binary tree by replacing each vertex $v$ by a triangle on vertices $r_v$, $x_v$ and $y_v$. Let $r_v$ be connected to the parent of $v$ and let $x_v$ and $y_v$ be connected to the left and right subtree of $v$. An example of a treegadget with $8$ leaves is shown in Figure \ref{fig:tree_gadget_lemmas}. If vertex $v$ is the root of the tree, then $r_v$ is named the \emph{root} of the treegadget. If $v$ does not have a left subtree, then $x_v$ is a \emph{leaf} of this gadget, similarly, if $v$ does not have a right subtree then we refer to $y_v$ as a leaf of the gadget. Let the \emph{height} of a treegadget be equal to the height of its corresponding binary tree.
\end{mydef}

It is easy to see that a treegadget is $3$-colorable. The important property of this gadget is that if there is a color that does not appear on any leaf in a proper $3$-coloring, then this must be the color of the root. See Figure \ref{fig:tree_gadget_lemma_1} for an illustration.

\begin{mylem}
\label{lem:must_color_root_c}
Let $T$ be a treegadget with root $r$ and let $c \colon V(T) \to \{1,2,3\}$ be a proper $3$-coloring of $T$. If  $k \in \{1,2,3\}$ such that $c(v) \neq k$ for every leaf $v$ of T, then $c(r) = k$.
\end{mylem}
\begin{proof}
This will be proven using induction on the structure of a treegadget. For a single triangle, the result is obvious. Suppose we are given a treegadget of height $h$ and that the statement holds for all treegadgets of smaller height. Consider the top triangle $r, x, y$ where $r$ is the root. Then, by the induction hypothesis, the roots of the left and right subtree are colored using $k$. Hence $x$ and $y$ do not use color $k$. Since~$x,y,r$ is a triangle,~$r$ has color $k$ in the $3$-coloring.
\end{proof}

The following lemma will be used in the correctness proof of the cross-composition to argue that the existence of a single \yes-input is sufficient for $4$-colorability of the entire graph.

\begin{mylem}
\label{lem:can_color_root_c}
Let $T$ be a treegadget with leaves $L\subseteq V(T)$ and root $r$. Any $3$-coloring $c' \colon L \to \{1,2,3\}$ that is proper on $T[L]$ can be extended to a proper $3$-coloring of $T$. If there is a leaf $v\in L$ such that $c'(v)=i$, then such an extension exists with $c(r) \neq i$.
\end{mylem}
\begin{proof}
We will prove this by induction on the height of the treegadget. For a single triangle, the result is obvious. Suppose the lemma is true for all treegadgets up to height $h-1$ and we are given a treegadget of height $h$ with root triangle $r,x,y$ and with coloring of the leaves $c'$. Let one of the leaves be colored using $i$. Without loss of generality assume this leaf is in the left subtree, which is connected to $x$. By the induction hypothesis, we can extend the coloring restricted to the leaves of the left subtree to a proper $3$-coloring of the left subtree such that $c(r_1) \neq i$. We assign color $i$ to $x$.
Since~$c'$ restricted to the leaves in the right subtree is a proper $3$-coloring of the leaves in the right subtree, by induction we can extend that coloring to a proper $3$-coloring of the right subtree. Suppose the root of this subtree gets color $j \in \{1,2,3\}$. We now color $y$ with a color $k \in \{1,2,3\}\setminus \{i,j\}$, which must exist. Finally, choose $c(r) \in \{1,2,3\} \setminus \{i,k\}$. By definition, the vertices $r$, $y$, and $x$ are now assigned a different color. Both $x$ and $y$ have a different color than the root of their corresponding subtree, thereby $c$ is a proper coloring. We obtain that the defined coloring $c$ is a proper coloring extending $c'$ with $c(r)\neq i$.
\end{proof}

\begin{mydef}
A \emph{triangular gadget} is a graph on $12$ vertices depicted in Figure \ref{fig:triangular_gadget}. Vertices $u,v$, and $w$ are the \emph{corners} of the gadget, all other vertices are referred to as \emph{inner vertices}.
\end{mydef}

It is easy to see that a triangular gadget is always $3$-colorable in such a way that every corner gets a different color.  Moreover, we make the following observation.

\begin{observation}
\label{lem:triangular_gadget}
Let $G$ be a triangular gadget with corners $u$,$v$ and $w$ and let $c \colon V(G) \to \{1,2,3\}$ be a proper $3$-coloring of $G$. Then $c(v) \neq c(u) \neq c(w) \neq c(v)$. Furthermore, every partial coloring that assigns distinct colors to the three corners of a triangular gadget can be extended to a proper $3$-coloring of the entire gadget.
\end{observation}

Having presented all the gadgets we use in our construction, we now define the source problem for the cross-composition. It is a variant of the problem that was used to prove kernel lower bounds for \textsc{Chromatic Number} parameterized by vertex cover~\cite{BodlaenderJK14}.

\defproblem{\problem{$2$-$3$-Coloring with Triangle Split Decomposition}}
{A graph $G$ with a partition of its vertex set into $X \cup Y$ such that $G[X]$ is an edgeless graph and $G[Y]$ is a disjoint union of triangles.}
{Is there a proper $3$-coloring $c:V(G) \to \{1,2,3\}$ of $G$, such that $c(x) \in \{1,2\}$ for all $x \in X$? We will refer to such a coloring as a \emph{2-3-coloring} of $G$.}

\begin{mylem}
\label{lem:NP_complete_coloring}
\problem{$2$-$3$-Coloring with Triangle Split Decomposition} is NP-complete.
\end{mylem}
\begin{proof}
It is easy to verify the problem is in NP. We will show that it is NP-hard by giving a reduction from \problem{$3$-nae-sat}, which is known to be NP-complete \cite[{[LO3] p. 259}]{NPComplete}.
Suppose we are given formula $F = C_1 \wedge C_2 \wedge \ldots \wedge C_m$ over set of variables $U$. Construct graph $G$ in the following way. For every variable $x \in U$, construct a gadget as depicted in Figure \ref{fig:NP_complete_variable}. For every clause $C_i$, construct a gadget as depicted in Figure \ref{fig:NP_complete_gadget}. Let $C_i = (\ell_1 \vee \ell_2 \vee \ell_3)$ for $i \in [m]$, connect vertex $\ell_j$ for $j \in \{1,2,3\}$ to vertex $v_{j}$ in gadget $C_i$ in $G$.

\begin{figure}[t]
\centering
\subfloat[Sub figure 1 list of figures text][Gadget for a variable]{
\includegraphics{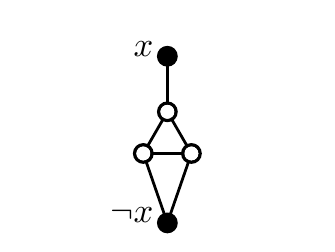}
\label{fig:NP_complete_variable}}
\qquad
\subfloat[Sub figure 2 list of figures text][Gadget for a clause]{
\includegraphics{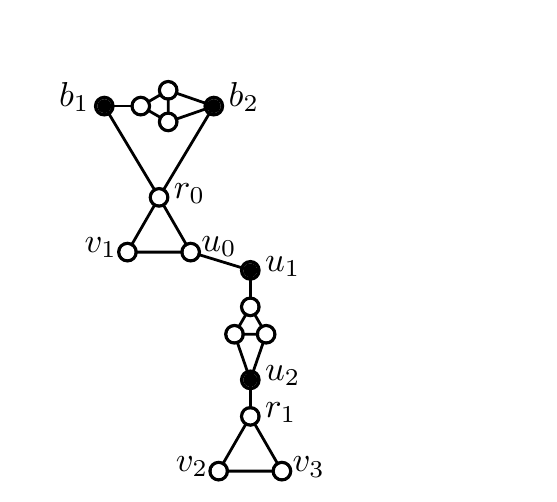}
\label{fig:NP_complete_gadget}}
\caption{The gadgets constructed for the clauses and variables of $F$.}
\label{fig:NP_complete_gadgets}
\end{figure}

It is easy to verify that $G$ has a triangle split decomposition. In Figure \ref{fig:NP_complete_gadgets}, triangles are shown with white vertices and the independent set is shown in black.

Suppose $G$ is $2$-$3$-colorable with color function $c:V(G) \rightarrow\{1,2,3\}$ and let $c(v) \in \{1,2\}$ for all $v$ in the independent set. Note that in each of the pairs $\{x, \neg x\}$, $\{b_1,b_2\}$, and $\{u_1, u_2\}$ the two vertices have distinct colors in any proper $2$-$3$-coloring of $G$. To satisfy $F$, let $x = true$ if and only if $c(x) = 2$. To show that this results in a satisfying assignment, consider any clause $C_i$ for $i \in [m]$. Note that $c(x) = 2 \Leftrightarrow c(\neg x) = 1$.
Since $c(b_1) \neq c(b_2)$ and $c(b_1),c(b_2) \in \{1,2\}$ we obtain $c(r_0) = 3$. Therefore, $v_1$ and $u_0$ are colored using colors $1$ and~$2$.

Suppose $c(v_1) = 1$. Thereby, $c(\ell_1) = 2$, implying the first literal of $C_i$ is set to $true$.  By $c(u_0) = 2$, we know $c(u_1) = 1$ and $c(u_2) = 2$. Thereby, $c(r_1) \neq 2$, so either $c(v_2) = 2$ or $c(v_3) = 2$. If $c(v_2) = 2$, then  $c(\ell_2) = 1$ which implies that literal $\ell_2$ is $false$ in $C_i$. Similarly, if $c(v_3) = 2$, then  $c(\ell_3) = 1$ which implies that literal $\ell_3$ is $false$ in $C_i$. In both cases it follows that clause $C_i$ is NAE-satisfied.

When $c(v_1) = 2$, we can use the same argument with the colors $1$ and $2$ swapped, to show that $\ell_1$ is $false$ in $C_i$ and $\ell_2$ or $\ell_3$ is $true$, which implies that $C_i$ is NAE-satisfied.

Suppose $F$ is a \yes-instance, with satisfying truth assignment $S$. Define color function $c:V(G)\rightarrow \{1,2,3\}$ as $c(x) := 1$ and $c(\neg x) := 2$ if $x$ is set to false in $S$, define $c(x) := 2$ and $c(\neg x) := 1$ otherwise. Color the remainder of the variable gadgets consistently. We now need to show how to color the clause gadgets. Consider any clause $C_i = (\ell_1\vee \ell_2\vee \ell_3)$. At least one of the literals is true and one is set to false, by symmetry we only consider four cases. The corresponding colorings are depicted in Figure \ref{fig:NP_complete_gadget_cases}, where red corresponds to $1$, green corresponds to $2$ and blue corresponds to color $3$. It is easy to verify that this leads to a proper $3$-coloring that only uses colors $1$ and $2$ on vertices in the independent set.
\begin{figure}[t]
\begin{center}
\includegraphics{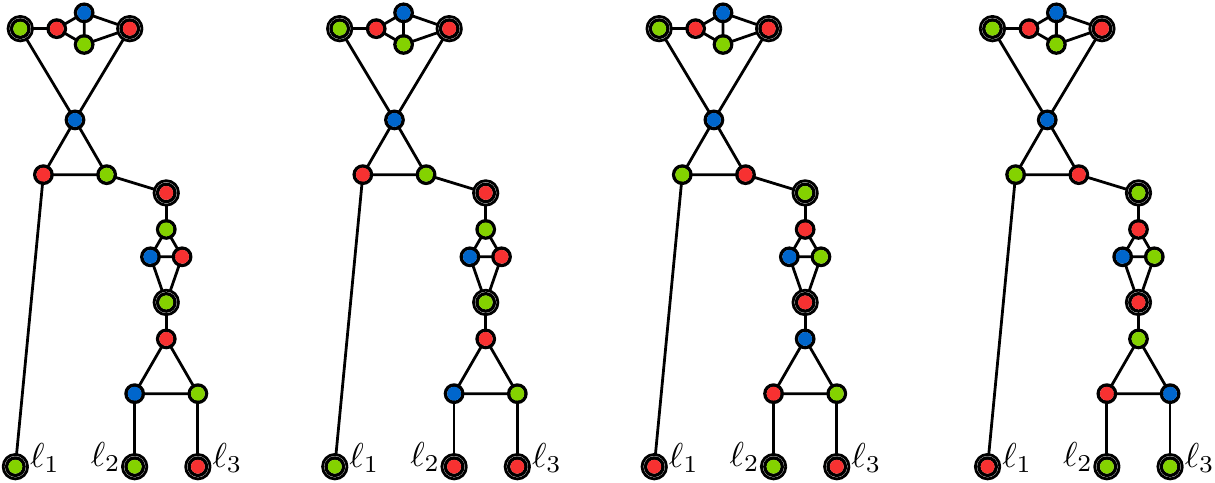}
\end{center}
\caption{Valid colorings of a clause gadget, depending on the coloring of the literals $\ell_1,\ldots,\ell_3$. Note that if the roles of $\ell_2$ and $\ell_3$ are exactly reversed, you can just exchange colors between their parents to get a proper coloring for that situation.}
\label{fig:NP_complete_gadget_cases}
\end{figure}
\end{proof}

\begin{mythm} \label{thm:fourcoloring:lowerbound}
\FourColoring parameterized by the number of vertices $n$ does not have a generalized kernel of size $\Oh(n^{2-\varepsilon})$ for any $\varepsilon > 0$, unless \containment.
\end{mythm}
\begin{proof}
By Theorem~\ref{thm:cross_composition_LB} and Lemma~\ref{lem:NP_complete_coloring} it suffices to give a degree-2 cross-composition from the $2$-$3$-coloring problem defined above into \FourColoring parameterized by the number of vertices. For ease of presentation, we will actually give a cross-composition into the \textsc{$4$-List Coloring} problem, whose input consists of a graph~$G$ and a list function that assigns every vertex~$v \in V(G)$ a list~$L(v) \subseteq [4]$ of allowed colors. The question is whether there is a proper coloring of the graph in which every vertex is assigned a color from its list. The \textsc{$4$-List Coloring} reduces to the ordinary \FourColoring by a simple transformation that adds a $4$-clique to enforce the color lists, which will prove the theorem. For now, we focus on giving a cross-composition into \textsc{$4$-List Coloring}.

We start by defining a polynomial equivalence relation on inputs of \problem{$2$-$3$-Coloring with Triangle Split Decomposition}. Let two instances of \problem{$2$-$3$-Coloring with Triangle Split Decomposition} be equivalent under equivalence relation $\eqvr$ when they have the same number of triangles and the independent sets also have the same size. It is easy to see that $\eqvr$ is a polynomial equivalence relation. By duplicating one of the inputs, we can ensure that the number of inputs to the cross-composition is an even power of two; this does not change the value of \textsc{or}, and increases the total input size by at most a factor four. We will therefore assume that the input consists of~$t$ instances of \problem{$2$-$3$-Coloring with Triangle Split Decomposition} such that~$t = 2^{2i}$ for some integer~$i$, implying that~$\sqrt{t}$ and~$\log \sqrt{t}$ are integers. Let $\q := \sqrt{t}$. Enumerate the instances as $X_{i,j}$ for $1\leq i,j \leq \q$. Each input $X_{i,j}$ consists of a graph $G_{i,j}$ and a partition of its vertex set into sets $U$ and $V$, such that $U$ is an independent set of size $m$ and $G_{i,j}[V]$ consists of $n$ vertex-disjoint triangles. Enumerate the vertices in $U$ and $V$  as $u_1,\ldots, u_m$ and $v_1, \ldots, v_{3n}$, such that vertices $v_{3\ell-2}, v_{3\ell-1}$ and $v_{3\ell}$ form a triangle, for $\ell \in  [n]$. We will create an instance $G'$ of the \problem{$4$-List-Coloring} problem, which consists of a graph~$G'$ and a list function~$L$ that assigns each vertex a subset of the color palette~$\{x,y,z,a\}$. 
Refer to Figure \ref{fig:4-coloring} for a sketch of $G'$.

\begin{figure}[t]
\begin{center}
\includegraphics{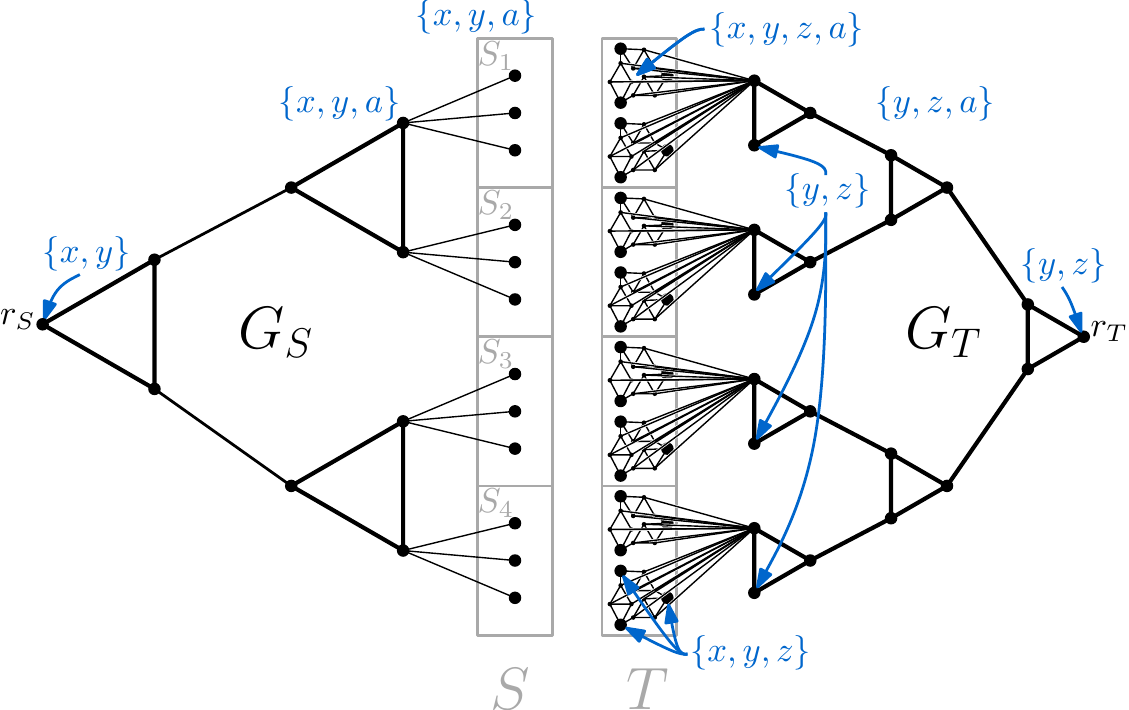}
\end{center}
\caption{The graph $G'$ for $\q = 4$, $m = 3$ and $n=2$. Edges between vertices in $S$ and $T$ are left out for simplicity.}
\label{fig:4-coloring}
\end{figure}

\begin{enumerate}
	\item Initialize~$G'$ as the graph containing $\q$ sets of $m$ vertices each, called $S_i$ for $i \in [\q]$. Label the vertices in each of these sets as $s^i_\ell$ for $i \in [\q]$, $\ell \in [m]$ and let $L(s^i_\ell) := \{x,y,a\}$.
	\item Add $\q$ sets of $n$ triangular gadgets each, labeled $T_j$ for $j \in [\q]$. Label the corner vertices in $T_j$ as $t_\ell^j$ for $\ell \in [3n]$, such that vertices $t_{3\ell-2}^j, t_{3\ell-1}^j$ and $t_{3\ell}^j$ are the corner vertices of one of the gadgets for $\ell \in [n]$. Let $L(t_\ell^j) := \{x,y,z\}$ and for any inner vertex $v$ of a triangular gadget, let $L(v) := \{x,y,z,a\}$.
	\item Connect vertex $s^i_k$ to vertex $t^j_\ell$ if in graph $G_{i,j}$ vertex $u_k$ is connected to $v_\ell$, for $k \in [m]$ and $\ell \in [3n]$. By this construction, the subgraph of $G'$ induced by $S_i \cup T_j$ is isomorphic to the graph obtained from~$G_{i,j}$ by replacing each triangle with a triangular gadget.
	\item Add a treegadget $G_S$ with $\q$ leaves to~$G'$ and enumerate these leaves as $1,\ldots, \q$; recall that~$\q$ is a power of two. Connect the $i$'th leaf of $G_S$ to every vertex in $S_i$. Let the root of $G_S$ be $r_S$ and define $L(r_S) := \{x,y\}$. For every other vertex $v$ in $G_S$ let $L(v) := \{x,y,a\}$.
	\item Add a treegadget $G_T$ with $2\q$ leaves to~$G'$ and enumerate these leaves as $1, \ldots, 2\q$. For $j \in [\q]$, connect every inner vertex of a triangular gadget in group $T_j$ to leaf number $2j-1$ of $G_T$. For every leaf $v$ with an even index let $L(v) := \{y,z\}$ and let the root $r_T$ have list $L(r_T) := \{y,z\}$. For every other vertex $v$ of gadget $G_T$ let $L(v) := \{y,z,a\}$.
\end{enumerate}

\begin{claim}
The graph $G'$ is $4$-list-colorable $\Leftrightarrow$ some input instance $X_{i^*j^*}$ is $2$-$3$-colorable.
\end{claim}

\begin{proof}
$(\Rightarrow)$
Suppose we are given a $4$-list coloring $c$ for $G'$. By definition, $c(r_S) \neq a$. From Lemma \ref{lem:must_color_root_c} it follows that there is a leaf $v$ of $G_S$ such that $c(v) = a$. This leaf is connected to all vertices in some $S_{i^*}$, which implies that none of the vertices in $S_{i^*}$ are colored using $a$. Therefore all vertices in $S_{i^*}$ are colored using $x$ and $y$.
Similarly the gadget $G_T$ has at least one leaf $v$ such that $c(v) = a$, note that this must be a leaf with an odd index. Therefore there exists $T_{j^*}$ where all vertices are colored using $x$,$y$ or $z$. Thereby in $S_{i^*} \cup T_{j^*}$ only three colors are used, such that $S_{i^*}$ is colored using only two colors. Using Observation \ref{lem:triangular_gadget} and the fact that~$G'[S_{i^*} \cup T_{j^*}]$ is isomorphic to the graph obtained from~$G_{i^*,j^*}$ by replacing triangles by triangular gadgets, we conclude that $X_{i^*j^*}$ has a proper $2$-$3$-coloring.

$(\Leftarrow)$ Suppose $c \colon V(G_{i^*,j^*}) \to \{x,y,z\}$ is a proper $2$-$3$-coloring for $X_{i^*,j^*}$. We will construct a $4$-list coloring $c' \colon V(G') \to \{x,y,z,a\}$ for $G'$.
For $u_k$, $k\in [m]$ in instance $X_{i^*,j^*}$ let $c'(s^{i^*}_k) := c(u_k)$ and for $v_\ell$ for $\ell \in [3n]$ let $c'(t^{j^*}_\ell) := c(v_\ell)$. Let $c'(s^i_\ell) := a$ for $i \neq i^*$ and $\ell \in [n]$, furthermore let $c'(t^j_\ell) := z$ for $j \neq j^*$ and $\ell \in [3m]$. For triangular gadgets in $T_{j^*}$ the coloring $c'$ defines all corners to have distinct colors; by Observation \ref{lem:triangular_gadget} we can color the inner vertices consistently using $\{x,y,z\}$. For $T_j$ with $j\in[\q]$ and $j \neq j^*$, the corners of triangular gadgets have color $z$ and we can now consistently color the inner vertices using $\{x,y,a\}$.

The leaf of gadget $G_S$ that is connected to $S_{i^*}$ can be colored using $a$. Every other leaf can use both $x$ and $y$, so we can properly $3$-color the leaves such that one leaf has color $a$. From Lemma \ref{lem:can_color_root_c} it follows that we can consistently $3$-color $G_S$ such that the root $r_S$ does not receive color $a$, as required by $L(r_S)$. Similarly, in triangular gadgets in $T_{j^*}$  the inner vertices do not have color $a$. As such, leaf $2j^*-1$ of $G_T$ can be colored using $a$ and we color leaf $2j^*$ with $y$. For $j \in [\q]$ with $j \neq j^*$ color leaf $2j - 1$ with $z$ and leaf $2j$ using $y$. Now the leaves of $G_T$ are properly $3$-colored and one is colored $a$. It follows from Observation~\ref{lem:triangular_gadget} that we can color $G_T$ such that the root is not colored $a$. This completes the $4$-list coloring of $G'$.
\end{proof}

The claim shows that the construction serves as a cross-composition into \problem{$4$-List Coloring}. To prove the theorem, we add four new vertices to simulate the list function. Add a clique on $4$ vertices $\{x,y,z,a\}$. If for any vertex $v$ in $G'$, some color is not contained in $L(v)$, connect $v$ to the vertex corresponding to this color. As proper colorings of the resulting graph correspond to proper list colorings of~$G'$, the resulting graph is $4$-colorable if and only if there is a \yes-instance among the inputs. It remains to bound the parameter of the problem, i.e., the number of vertices. Observe that a treegadget has at least as many leaves as its corresponding binary tree, therefore the graph $G'$ has at most $12m\q+n\q + 6\q + 12\q + 4 = \Oh(\q \cdot (m+n)) = \Oh(\sqrt{t}\max|X_{i,j}|)$ vertices. Theorem~\ref{thm:fourcoloring:lowerbound} now follows from Theorem \ref{thm:cross_composition_LB} and Lemma~\ref{lem:NP_complete_coloring}.
\end{proof} 

\section{Hamiltonian cycle}
\label{sec:hamiltonian_cycle}
In this section we prove a sparsification lower bound for \textsc{Hamiltonian Cycle} and its directed variant by giving a degree-$2$ cross-composition. The starting problem is \problem{Hamiltonian $s-t$ path on bipartite graphs}.

\defproblem{\problem{Hamiltonian $s-t$ path on bipartite graphs}}{An undirected bipartite graph $G$ with partite sets~$A$ and~$B$ such that $|B| = n = |A| + 1$, together with two distinguished vertices $b_1$ and $b_n$ that have degree $1$.}{Does $G$ have a Hamiltonian path from $b_1$ to $b_n$?}

It is known that Hamiltonian path is NP-complete on bipartite graphs \cite[{[GT39] (2)}]{NPComplete} and it is easy to see that is remains NP-complete when fixing a degree $1$ start and endpoint.

\begin{mythm}
\problem{(Directed) Hamiltonian Cycle} parameterized by the number of vertices $n$ does not have a generalized kernel of size $\Oh(n^{2 - \varepsilon})$ for any $\varepsilon > 0$, unless \containment.
\end{mythm}
\begin{proof}
By a suitable choice of polynomial equivalence relation, and by padding the number of inputs, it suffices to give a cross-composition from the $s-t$ problem on bipartite graphs when the input consists of~$t$ instances~$X_{i,j}$ for~$i, j \in [\sqrt t]$ (i.e.,~$\sqrt t$ is an integer), such that each instance~$X_{i,j}$ encodes a bipartite graph~$G_{i,j}$ with partite sets~$A^*_{i,j}$ and~$B^*_{i,j}$ with $|A^*_{i,j}| = m$ and $|B^*_{i,j}| = n = m+1$, for some~$m \in \mathbb{N}$. For each instance, label all elements in $A^*_{i,j}$ as $a^*_1, \ldots, a^*_m$ and all elements in $B^*_{i,j}$ as $b^*_1, \ldots, b^*_n$ such that $b^*_1$ and $b^*_n$ have degree $1$.

The construction makes extensive use of the path gadget depicted in Figure \ref{fig:domino_gadget}. Observe that if $G'$ contains a path gadget as an induced subgraph, while the remainder of the graph only connects to its terminals $\textsc{in}^0$ and $\textsc{in}^1$, then any Hamiltonian cycle in $G'$ traverses the path gadget in one of the two ways depicted in Figure \ref{fig:domino_gadget}. We create an instance $G'$ of \problem{Directed Hamiltonian Cycle} that acts as the logical \textsc{or} of the inputs.

\begin{enumerate}
\item \label{hc:step:A} First of all construct $\sqrt{t}$ groups of $m$ path gadgets each. Refer to these groups as $A_i$, for $i \in [\sqrt{t}]$, and label the gadgets within group $A_i$ as $a_1^i, \ldots, a_m^i$. Let the union of all created sets $A_i$ be named $A$.
 \label{hc:step:B} Similarly, construct $\sqrt{t}$ groups of $n$ path gadgets each. Refer to these groups as $B_j$, for $j \in [\sqrt{t}]$, and label the gadgets within group $B_j$ as $b_1^j, \ldots, b_n^j$.  Let $B$ be the union of all $B_j$ for $j \in [\sqrt{t}]$.
\item \label{hc:step:connectAB} For every input instance $X_{i,j}$, for each edge $\{a^*_k, b^*_\ell\}$ in $X_{i,j}$ with $k \in [m]$, $\ell \in [n]$,
add an arc from $\textsc{in}^0$ of $a_k^i$ to $\textsc{in}^1$ of $b_\ell^j$
and an arc from $\textsc{in}^0$ of $b_\ell^j$ to $\textsc{in}^1$ of $a_k^i$.
\end{enumerate}
If some $X_{i,j}$ has a Hamiltonian  $s-t$ path, it can be mimicked by the combination of $A_i$ and $B_j$, where for each vertex in $X_{i,j}$ we traverse its path gadget in $G'$, following Path $1$. The following construction steps are needed to extend such a path to a Hamiltonian cycle in $G'$. 

\begin{enumerate}[resume]
\item \label{hc:step:Apath} Add an arc from the $\textsc{in}^1$ terminal of $a_\ell^i$  to the $\textsc{in}^0$ terminal of $a_{\ell+1}^i$ for all $\ell \in[m-1]$ and all $i \in [\sqrt{t}]$.
\label{hc:step:Bpath} Similarly add an arc from the $\textsc{in}^1$ terminal of $b_\ell^i$ to the $\textsc{in}^0$ of $b_{\ell+1}^i$ for all $\ell \in [n-1]$ and all $i \in [\sqrt{t}]$.
\item Add a vertex $\textsc{start}$ and a vertex $\textsc{end}$ and the arc $(\textsc{end}, \textsc{start})$.
\item Let $r := \sqrt{t}-1$, add $2r$ tuples of vertices, $x_i, y_i$ for $i \in [2r]$ and connect $\textsc{start}$ to $x_1$.
 Furthermore, add the arcs $(y_i, x_{i+1})$ for $i \in [2r-1]$.
\item \label{hc:step:x_to_A} For $i \leq r$ we add arcs from $x_i$ to the $\textsc{in}^0$ terminal of the gadgets $a_1^j, j\in[\sqrt{t}]$. Furthermore we add an arc from $\textsc{in}^1$ of $a_m^j$ to $y_i$ for all $j\in [\sqrt{t}]$ and $i \in [r]$.
When $i > r$ add arcs from $x_i$ to the $\textsc{in}^0$ terminal of $b_1^j$ for $j \in [\sqrt{t}]$ and connect $\textsc{in}^1$ of $b_n^j$ to $y_i$.
\item Add a vertex $\textsc{next}$ and the arc $(y_{2r}, \textsc{next})$ and an arc from $\textsc{next}$ to the $\textsc{in}^1$ terminal of all gadgets $b_1^j$ for $j \in [\sqrt{t}]$.
\item Furthermore, add arcs from $\textsc{in}^0$ of all gadgets $b_n^j$ to $\textsc{end}$ for $j \in [\sqrt{t}]$. So for each $B_j$, exactly one vertex has an outgoing arc to $\textsc{end}$ and one has an incoming arc from $\textsc{next}$.
\end{enumerate}

\begin{figure}[t]
\centering
\subfloat[Domino gadget][A path gadget.]{
\raisebox{2cm}{\includegraphics{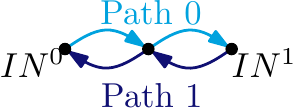}}
\label{fig:domino_gadget}}
\subfloat[Instance for directed Hamiltonian cycle][The general structure of the created graph, when given $4$ inputs with $n = 3$ and $m = 4$.]{
\includegraphics{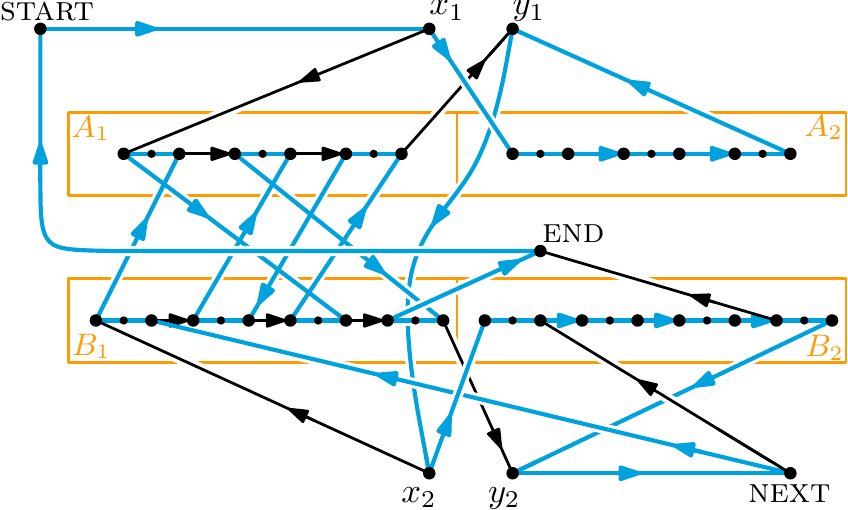}
\label{fig:Hamiltonian_decision}}
\caption{Illustrations for the lower bound for \textsc{Hamiltonian Cycle}.}
\label{fig:Hamiltonian_figures}
\end{figure}

This completes the construction of $G'$. A sketch of $G'$ is shown in Figure \ref{fig:Hamiltonian_decision}.
In order to prove that the created graph $G'$ acts as a logical \textsc{or} of the given input instances, we first establish a number of auxiliary lemmas.

\begin{mylem}
\label{lem:path_0_or_1}
Any Hamiltonian cycle in $G'$ traverses any path gadget in $G'$ via directed Path 0 or Path 1, as shown in Figure \ref{fig:domino_gadget}.
\end{mylem}
\begin{proof}
\label{plem:path_0_or_1}
Any Hamiltonian cycle in $G'$ should visit the center vertex of the path gadget. Since $\textsc{in}^0$ and $\textsc{in}^1$ are its only two neighbors in $G'$, the only option is to visit them consecutively,  Path $0$ and Path $1$ are the only two options to do this.
\end{proof}

\begin{mylem}
\label{lem:traverse_sequence_i}
When any Hamiltonian cycle in $G'$ enters path gadget $a_1^i$ at $\textsc{in}_0$ for some~$i \in [\sqrt{t}]$, the cycle then visits the gadgets $a_2^i, a_3^i,\ldots,a_m^i$ in order without visiting other vertices in between. Similarly, if any Hamiltonian cycle in $G'$ enters path gadget $b_1^j$ at $\textsc{in}_0$, the cycle then visits the gadgets $b_2^j, b_3^j,\ldots,b_n^j$ in order without visiting other vertices in between.
\end{mylem}
\begin{proof}
\phantomsection
\label{plem:traverse_sequence_i}


Consider a Hamiltonian cycle in~$G'$ that enters path gadget~$a_1^i$ at $\textsc{in}_0$. By Lemma~\ref{lem:path_0_or_1} the cycle follows Path $0$ and continues to the $\textsc{in}^1$ terminal of the path gadget. Since that terminal has only one out-neighbor outside the gadget, which leads to the $\textsc{in}_0$ terminal of~$a_2^i$, it follows that the cycle continues to that path gadget. As the adjacency structure around the other path gadgets is similar, the lemma follows by repeating this argument. The proof when entering group $B_j$ at the vertex $\textsc{in}_0$ of $b_1^j$ is equivalent.
\end{proof}

\begin{mylem}
\label{lem:all_but_one}
Let $C$ be a directed Hamiltonian cycle in $G'$, such that its first arc is $(\textsc{start}, x_1)$. There are indices $i^*,j^* \in[\sqrt{t}]$ such that subpath $C_{x_1,y_{2r}}$ of the cycle between $x_1$ and $y_{2r}$ contains exactly the vertices
 \[\overline{A_{i^*}} \cup \overline{B_{j^*}} \cup \{x_i,y_i \mid i \in [2r]\}\]
where $\overline{A_{i^*}}$ contains all vertices of all gadgets in $A_i$ for $i \neq i^*$, and similarly $\overline{B_{j^*}}$ contains all vertices of all gadgets in $B_j$ for $j \neq j^*$.
\end{mylem}
\begin{proof}
We will first show that when the cycle reaches any $x_i$ for $i \in [r]$, it traverses exactly one group $A_\ell$ with $\ell \in [r+1]$ and continues to $y_j$ and  $x_{j+1}$ for some $j \in [r]$, without visiting other vertices in between. Similarly, when the cycle reaches any $x_i$ for $r < i \leq 2r$, it traverses exactly one group $B_\ell$ with $\ell \in [r+1]$ and continues to $y_j$ for some $r < j \leq 2r$. For $j < 2r$, the cycle then continues to $x_{j+1}$, for $j = 2r$ the cycle reached $y_{2r}$, which is the last vertex of this subpath.

By Step \ref{hc:step:x_to_A} in the construction, all outgoing arcs of any $x_i$ for $i \in [r]$ lead to gadgets $a_1^\ell$ for some $\ell \in[\sqrt{t}]$. So for any $x_i$ in the cycle there must be a unique $\ell \in [\sqrt{t}]$ such that the arc from $x_i$ to the $\textsc{in}^0$ terminal of $a_1^{\ell}$ is in $C$. By Lemma \ref{lem:traverse_sequence_i} the cycle visits all vertices in $A_\ell$, and no other vertices, before reaching gadget  $a_m^{\ell}$, which is traversed by Path $0$ to get to $\textsc{in}^1$ of this gadget. The only neighbors of $\textsc{in}^1$ of gadget $a_m^{\ell}$ lying outside this gadget are of type $y_j$ for $j \in [r]$. As such, the cycle must visit some $y_j$ next, and its only outgoing arc goes to $x_{j+1}$.

The proof for $i > r$ is similar.
As such, visiting $x_i$ for $i \in [r]$ results in visiting all vertices of exactly one group in $A$ before continuing via $y_j$ to some $x_{j+1}$ without visiting any vertices in between. Visiting $x_i$ for $r < i \leq 2r$ results in visiting all vertices of exactly one group in $B$ and returning via $y_j$ to either the end of the subpath ($j=2r$) or some $x_{j+1}$.

Every vertex $x_i$ for $i \in [2r]$ must be visited by $C$, it remains to show that it is visited in subpath $C_{x_1, y_{2r}}$. Suppose there exists an $x_i$ for $i \in [2r]$ such that $x_i$ is not visited in the subpath from $x_1$ to $y_{2r}$.  As we have seen above, visiting some $x_i$ results in visiting all vertices in some group in $A$ or $B$, continued by visiting some $y_j$ for $j \in [2r]$. Note that no other vertices are visited in between.
Hereby, $y_j$ is not in subpath $C_{x_1, y_{2r}}$. This implies $j \neq 2r$ and thus the next vertex in the cycle is $x_{j+1}$. So, for $x_i$ not in subpath $C_{x_1, y_{2r}}$, one can find a new vertex $x_{j+1}$ (where $j+1 \neq i$), such that $x_{j+1}$ is also not in subpath $C_{x_1, y_{2r}}$. Note that we can not create a loop, by visiting a vertex $x_i$ seen earlier, as this would not yield a Hamiltonian cycle in $G'$. For example, the vertex \textsc{start} would never be visited. This is however a contradiction since we only have finitely many vertices $x_i$.

Thus in subpath $C_{x_1, y_{2r}}$, exactly $r$ groups of $A$ are visited and exactly $r$ groups of $B$ are visited, and no other vertices than specified. This leaves exactly one group $A_{i^*}$ and one group $B_{j^*}$ unvisited in $C_{x_1,y_{2r}}$.
\end{proof}

In Step \ref{hc:step:x_to_A} we create a selection mechanism that leaves one group in $A$ and one in $B$ unvisited. The following lemma formalizes this idea.

\begin{mylem}
\label{lem:next_to_end}
Let $C$ be a Hamiltonian cycle in $G'$, such that its first arc is $(\textsc{start}, x_1)$. Let $i^*$ and $j^*$ satisfy the conditions of Lemma \ref{lem:all_but_one}. Then cycle $C$ visits $b_1^{j^*}$ before $b_n^{j^*}$. Moreover, the subpath of the cycle $C_{b_1^{j^*},b_n^{j^*}}$ between terminal $\textsc{in}^1$ of $b_1^{j^*}$ and $\textsc{in}^0$ of $b_n^{j^*}$ (inclusive) contains all vertices of the gadgets in $A_{i^*}$ and $B_{j^*}$ and no others.
\end{mylem}
\begin{proof}
\phantomsection
\label{plem:next_to_end}

Vertex $\textsc{next}$ is visited directly after $y_{2r}$, since it is the only out-neighbor of~$y_{2r}$. Furthermore, the arc from $\textsc{next}$ to gadget $b_1^\ell$ must be in the cycle for some $\ell \in [\sqrt{t}]$, since $\textsc{next}$ only has outgoing arcs of this type. By Lemma $\ref{lem:all_but_one}$, all gadgets in all $B_j$ for $j \neq j^*$ are visited in the path from $x_1$ to $y_{2r}$, and thus should not be visited after vertex $\textsc{next}$. Therefore, the arc from $\textsc{next}$ to gadget $b_1^{j^*}$ is in the cycle, which also implies that $b_1^{j^*}$ is visited before $b_n^{j^*}$.

It is easy to see that $(\textsc{end},\textsc{start})$ is the last arc in $C$. By considering the incoming arcs of $\textsc{end}$ it follows that some arc from terminal $\textsc{in}^0$ of $b_n^{\ell}$ to $\textsc{end}$ for $\ell \in [\sqrt{t}]$ is in the cycle. Since the vertices in gadgets $b_n^\ell$ for $\ell \neq j^*$ are already visited in $C_{x_1, y_{2r}}$ by Lemma \ref{lem:all_but_one}, it follows that $(b_n^{j^*} , \textsc{end})$ is in $C$.

By Lemma \ref{lem:all_but_one}, none of the terminals of gadgets in $A_{i^*}$ and $B_{j^*}$ are visited in the subpath $C_{x_1,y_{2r}}$ or equivalently in the subpath $C_{\textsc{start},\textsc{next}}$. Since $C$ is a Hamiltonian cycle these vertices must therefore be visited in $C_{\textsc{next},\textsc{start}}$, which is equivalent to saying that $C_{b_1^{j^*}, b_n^{j^*}}$ must contain all vertices in $A_{i^*}\cup B_{j^*}$. It is easy to see that this subpath cannot contain any other vertices, as all other vertices are  present in $C_{\textsc{start},\textsc{next}}$ or $C_{\textsc{end},\textsc{start}}$.
\end{proof}

Using the lemmas above, we can now prove that $G'$ has a Hamiltonian cycle if and only if one of the instances has a Hamiltonian path.

\begin{mylem}
\label{lem:Hamiltonian_cycle}
Graph $G'$ has a directed Hamiltonian cycle if and only if at least one of the instances $X_{i,j}$ has a Hamiltonian $s-t$-path.
\end{mylem}
\begin{proof}

($\Leftarrow$) Suppose $G'$ has a Hamiltonian cycle $C$. By Lemma \ref{lem:next_to_end} there exist $i^*, j^* \in [\sqrt{t}]$ such that the subpath of $C$ from gadget $b_1^{j^*}$ to $b_n^{j^*}$ visits exactly the gadgets in $A_{i^*} \cup B_{j^*}$. Since gadget $b_1^{j^*}$ is entered at terminal $\textsc{in}^1$, it is easy to see that all gadgets are traversed using Path $1$. We now construct a Hamiltonian path $P$ for instance $X_{i^*,j^*}$. Let $\{a^*_k(i^*,j^*), b^*_\ell(i^*,j^*)\} \in P$ if  the arc from $\textsc{in}^0$ of $a_k^{i^*}$ to $\textsc{in}^1$ of $b_\ell^{j^*}$ is in $C$.
Similarly let $\{b^*_k(i^*,j^*), a^*_\ell(i^*,j^*)\} \in P$ if the arc from $\textsc{in}^0$ of $b_\ell^{j^*}$ to  $\textsc{in}^1$ of $a_k^{i^*}$ is in $C$, where $k \in [m]$ and $\ell \in [n]$.  Using that every gadget is visited exactly once via Path $1$ in $C$, we see that $C$ is a Hamiltonian path.

($\Rightarrow$) Suppose $X_{i^*,j^*}$ has a Hamiltonian $s-t$ path $P$. Then we create a Hamiltonian cycle $C$, for each vertex $a^*_\ell$ from instance $X_{i^*,j^*}$ in $P$ we add Path $1$ in path gadget $a_\ell^{i^*}$ to $C$ and for each vertex $b^*_\ell$  we add Path $1$ in path gadget $b_\ell^{j^*}$ to $C$. Let $P$ be ordered such that $b^*_1$ is its first vertex. Now if $a^*_k$ is followed by $b^*_\ell$ in $P$, the arc from terminal $\textsc{in}^0$ of $a^{i^*}_k$ to $\textsc{in}^1$ of $b^{j^*}_\ell$ is added to $C$. Similarly, if a vertex $b_\ell^*$ is followed by $a_k^*$ in $P$, the arc from terminal $\textsc{in}^0$ of $b^{j^*}_\ell$ to $\textsc{in}^1$ of $a^{i^*}_k$ will be added to $C$. Now the subpath $C_{b_1^{j^*},b_n^{j^*}}$ contains all terminals in all gadgets in $A_{i^*} \cup B_{j^*}$.

From $b_n^{j^*}$ the cycle goes to $\textsc{end}$, then to $\textsc{start}$ and to $x_1$. To visit all groups $A_i$ for $i \neq i^*$ and $B_j$ for $j \neq j^*$, do the following.
\begin{itemize}
\item From $x_i$ where $i < i^*$, the cycle continues to gadgets $a_1^i$, then to $a_2^i, a_3^i, \ldots, a_m^i$ following Path $0$, and continue to $y_i, x_{i+1}$.
\item From $x_i$ where $i^* \leq i \leq r$ it goes to $a_1^{i+1},a_2^{i+1},\ldots, a_n^{i+1}$ and continues with $y_i,x_{i+1}$.
\item Similarly, from $x_i$ where $r \leq i < j^*$, go through gadgets $b_1^i,\ldots, b_n^i$ and continue to $y_i, x_{i+1}$.
\item From $x_i$ where $j^* \leq i \leq 2r$, go to gadgets $b_1^{i+1},\ldots,b_n^{i+1}$ and continue to $y_i$, for $i \neq 2r$ then add the arc $(y_i,x_{i+1})$.
\end{itemize}
From $y_{2r}$, continue to $\textsc{next}$, after which the arc $(\textsc{next}, b_1^{j^*})$ closes the cycle. By definition, no vertex is visited twice, so it remains to check that every vertex of $G'$ is in the cycle. For vertices $\textsc{start}, \textsc{next}, \textsc{end}$ and all vertices $x_i, y_i, z_i$ this is obvious. All vertices in $A_i$ and $B_j$ where $i \neq i^*$ and $j \neq j^*$ are in the cycle between some $x_\ell$ and $y_\ell$. All vertices in $A_{i^*}$ and $B_{j^*}$ are visited since $P$ was a Hamiltonian path on these vertices.
\end{proof}

The number of vertices of $G'$ is $3(m+n)\sqrt{t}+3 \cdot 2(\sqrt{t}-1) + 3 = \Oh(\sqrt{t} \cdot (m+n)) = \Oh(\sqrt{t} \cdot \max|X_{i,j}|)$. By with Lemma \ref{lem:Hamiltonian_cycle} the construction is a degree-$2$ cross-composition from \problem{Hamiltonian $s-t$-paths in Bipartite graphs} to \problem{Directed Hamiltonian cycle} parameterized by the number of vertices, proving the generalized kernel lower bound for the directed problem. Karp \cite{Karp72} gave a polynomial-time reduction that, given an $n$-vertex directed graph $G$, produces an undirected graph $G'$ with $3n$ vertices such that $G$ has a directed Hamiltonian cycle if and only if $G'$ has a Hamiltonian cycle. This is a linear parameter transformation  from \problem{Directed Hamiltonian cycle} to \problem{Hamiltonian cycle}. Since linear-parameter transformations transfer lower bounds~\cite{BodlaenderJK14,BodlaenderTY11}, we conclude that \problem{(Directed) Hamiltonian cycle} does not have a generalized kernel of size $\Oh(n^{2-\varepsilon})$ for any $\varepsilon > 0$.
\end{proof} 

\section{Dominating set}
\label{sec:dominating_set}
In this section we discuss the \textsc{Dominating Set} problem and its variants. Dom \emph{et al.}~\cite{Dom2014Incompressability} proved several kernelization lower bounds for the variant \textsc{Red-Blue Dominating Set}, which is the variant on bipartite (red/blue colored) graphs in which the goal is to dominate all the blue vertices by selecting a small subset of red vertices. Using ideas from their kernel lower bounds for the parameterization by either the number of red or the number of blue vertices, we prove sparsification lower bounds for \textsc{(Connected) Dominating Set}. Since we parameterize by the number of vertices, the same lower bounds apply to the dual problems \textsc{Nonblocker} and \textsc{Max Leaf Spanning Tree}.

We will prove these sparsification lower bounds using a degree-$2$ cross-composition, starting  from a variation of the \problem{Colored Red-Blue Dominating Set} problem (\problem{Col-RBDS}) as described by Dom \etal\ in \cite[page 13:6]{Dom2014Incompressability}.

\defproblem{\problem{Equal-Sized Colored Red/Blue Dominating Set (Eq-Col-RBDS)}}{A bipartite graph $G = (R \cup B,E)$, where  $R$ is partitioned into $k$ subsets  $R_1, \ldots, R_k$, such that $|R_1| = |R_2| = \ldots = |R_k|$.}{Is there a set $S \subseteq R$ such that for each $i \in [k]$ the set $S$ contains exactly one vertex of $R_i$ and every vertex in $B$ is adjacent to at least one vertex from $S$.}

We will think of the vertices in set $R_i$ as having color $i$. Hence the question is whether there is a set $S \subseteq R$ containing exactly one vertex of each color, such that every vertex in $B$ is adjacent to at least one vertex in $S$.

\begin{mylem}
\label{lem:col_RBDS_NP}
\problem{eq-Col-RBDS} is NP-complete.
\end{mylem}
\begin{proof}
 Dom \etal\ \cite{Dom2014Incompressability}  proved the NP-completeness of Colored RBDS without the constraint that all color sets have equal size. The NP-completeness for the equal-sized version follows from the fact that we may repeatedly add isolated vertices to classes $R_i$ that are too small, without changing the answer.
\end{proof}

Using this result, we can now give a degree-$2$ cross-composition and prove the following.

\begin{mythm}
\label{thm:dominating_set}
\problem{(Connected) Dominating Set}, \problem{Nonblocker}, and \problem{Max Leaf Spanning Tree} parameterized by the number of vertices $n$ do not have a generalized kernel of size $\Oh(n^{2-\varepsilon})$ for any $\varepsilon > 0$, unless \containment.
\end{mythm}

\begin{proof}
A graph has a nonblocker of size $k$ if and only if it has a dominating set of size $n-k$. Furthermore,
the \problem{Maximum Leaf Spanning Tree} problem is strongly related to \problem{Connected Dominating Set}.
The internal vertices of any spanning tree form a connected dominating set. Conversely, any connected dominating set contains a subtree spanning the dominating set, which -- by the domination property -- can be greedily extended to a spanning tree for the entire graph in which the remaining vertices are leaves. Hence a graph has a connected dominating set of size at most $k$ if and only if it has a spanning tree with at least $n-k$ leaves. Therefore we will show this result for \problem{(Connected) Dominating Set} only.

Define a polynomial equivalence relation $\eqvr$ on instances of \problem{eq-Col-RBDS} by first of all letting all instances where there is a vertex in $B$ of degree $0$ be in the same class, note that these are always no-instances.
Let $2$ instances $(G = (R\cup B),k)$ and $(G' = (R'\cup B'),k')$ of \problem{eq-Col-RBDS} be equivalent if $|R| = |R'|$ , $|B| = |B'|$ and $k = k'$. It is easy to see that $\eqvr$ indeed is a polynomial equivalence relation. 

Suppose we are given $t$ instances of \problem{eq-Col-RBDS}, such that $\sqrt{t}$ and $\log{\sqrt{t}} \in \mathbb{N}$ and such that all given instances are in the same equivalence class of $\eqvr$. Let $\q := \sqrt{t}$. If these instances are from the class where $B$ contains a vertex of degree $0$, output a constant size no-instance.

Otherwise, label the given instances  as $X_{i,j}$ with $i,j \in [\q]$. Let instance $X_{i,j}$ have graph $G_{i,j}$, which is bipartite with vertex set $R^*_{i,j} \cup B^*_{i,j}$. Let $|R^*_{i,j}| = m$ and $|B^*_{i,j}| = n$ and let $R^*_{i,j}$ be partitioned into $k$ color classes $R^{*p}_{i,j}$ for all $i,j \in [{\q}]$ and $p \in [k]$. Label all vertices in $R^{*p}_{i,j}$ as $r^*_{p,q}(i,j)$ with $p \in [k]$ and $q \in [m/k]$, which means that this vertex is the $q$'th vertex of color $p$ from instance $X_{i,j}$. Label vertices in $B^*_{i,j}$ as $b^*_1(i,j),\ldots,b^*_n(i,j)$ arbitrarily.
We now create an instance $(G,k)$ for \problem{Dominating Set} using the following steps.  A sketch of $G$ can be found in Figure \ref{fig:dominating_set}.
\begin{figure}
\begin{center}
\includegraphics{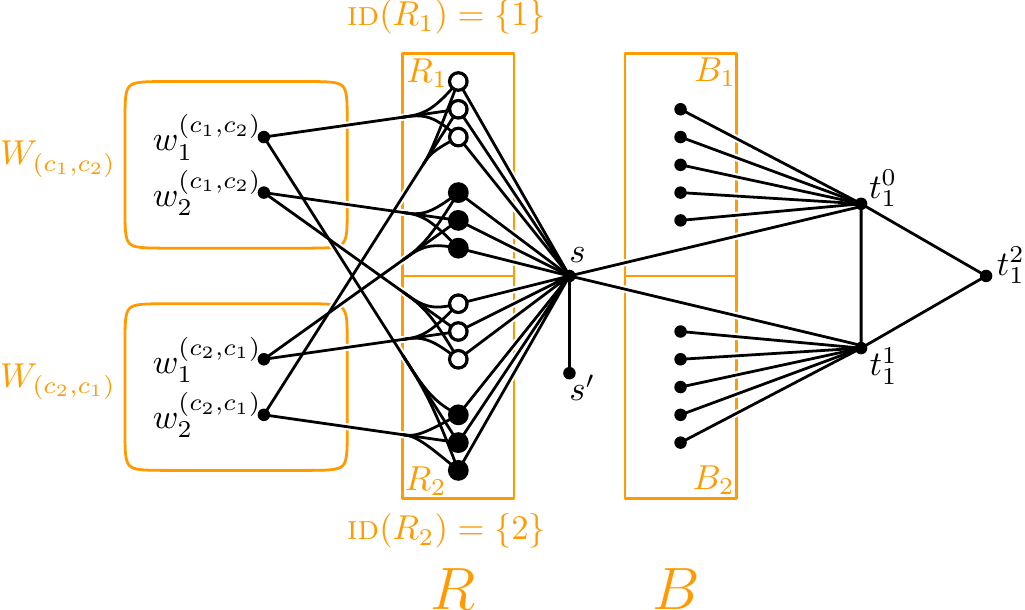}
\end{center}
\caption[Graph $G$ constructed in the proof of Theorem \ref{thm:dominating_set} for \problem{(Connected) Dominating Set}.]{A sketch of $G$, where $t'=2$, $m=6$ , $n=5$ and $k=2$. Thereby $K$ should be $5$ and $W_{(c_1,c_2)}$ should contain $10$ vertices. In this example we show the constructed graph when choosing $K=1$ for simplicity. We use the two colors $c_1$ and $c_2$, corresponding to white and black in the figure. Edges from $R$ to $B$ are left out for simplicity.}
\label{fig:dominating_set}
\end{figure}

\begin{enumerate}
\item \label{ds:step:A} Add vertices $r_{p,q}^i$ for $p \in [k], q \in [m/k]$ and $i \in [\q]$. The dominating set problem does not use colored instances, however we will remember the color of these vertices for simplicity. Let vertex $r_{p,q}^i$ have color $p$, for $i  \in [\q]$, $q \in [m/k]$ and $p \in [k]$.
Define $R_i := \{r_{p,q}^i \mid p \in [k], q\in[m/k]\}$ and let $R := \bigcup_{i\in [\q]} R_i$.
%
Give every set $R_i$ a unique identifier $\ID(R_i)$, which is a subset of $K := 2+k+\log \q$ numbers in the range $[2K]$.
\item \label{ds:step:B} Add vertices $b_\ell^j$ for $\ell \in [n]$ and $j \in[\q]$. Define $B_j$ and $B$ as $B_j := \{b_\ell^j \mid \ell \in [n]\}$ and $B := \bigcup_{j \in [\q]} B_j$.
\item \label{ds:step:connectAB} Add edges between the vertices $r_{p,q}^i$ and $b_\ell^j$ for $p \in [k], q \in [m/k]$ and $i,j \in [\q]$ if $r^*_{p,q}(i,j)$ is connected to $b^*_\ell(i,j)$ in instance $X_{i,j}$.
This ensures that the graph induced by $R_i \cup B_j$ is exactly $G_{ij,}$ and the coloring of vertices in $R_i$ matches the coloring of $R^*_{i,j}$.
\item \label{ds:step:s} Add vertices $s'$ and $s$ and edge $\{s',s\}$. Furthermore, add edges between $s$ and all vertices in $R$. The degree-$1$ vertex $s'$ ensures there is a minimum dominating set containing $s$, which covers all vertices in $R$ ``for free''.
\item \label{ds:step:W} In a similar way as given by Dom \etal\ in \cite{Dom2014Incompressability}, for every pair of colors $(c_1,c_2) \in \{1,\ldots,k\}\times\{1,\ldots,k\}$ with $c_1 \neq c_2$ we add a vertex set $W_{(c_1,c_2)}=\{w_1^{(c_1,c_2)},\ldots,w_{2K}^{(c_1,c_2)}\}$.
%
For $x \in [2K]$ connect $w_x^{(c_1,c_2)}$ to all vertices of color $c_1$ in $R_i$ if  $x \in \ID(R_i)$, otherwise connect $w_x^{(c_1,c_2)}$ to all vertices of color $c_2$ in $R_i$. This construction is used to choose which $R_i$ is part of a solvable input instance $X_{ij}$ for some $j \in [\q]$. This idea is formalized in Lemmas \ref{lem:one_from_each_color_from_A} and \ref{lem:all_in_one_square}.
\item \label{ds:step:T} Then, add $\log{\q}$ triangles, with vertices $\{t_\ell^0, t_\ell^1, t_\ell^2\}$ for $\ell \in[ \log{\q}]$. Connect $t_\ell^0$ to all vertices in $B_j$ if the $\ell$'th bit of $j$ equals $0$, connect $t_\ell^1$ to all vertices in $B_j$ if the $\ell$'th bit of $j$ equals $1$. Define $T$ to be the union of all these triangles. By choosing exactly one of the vertices $t_\ell^0$ or $t_\ell^1$ in a dominating set for each $\ell$, all groups $B_j$ except one are dominated automatically. The non-dominated one should then be part of a solvable input instance.
\item \label{ds:step:connected} Finally, add the edges $\{\{s,t_\ell^i\} \mid \ell \in [\log \q], i \in \{0,1\}\}$. This step ensures that every vertex in $T$ that is contained in the dominating set has $s$ as a neighbor in the dominating set, which implies that there is always a minimum dominating set that is connected. 
\end{enumerate}

We now make the following observations.

\begin{mylem}
\label{lem:no_from_B}
If $G$ has a dominating set $D$, then it also has a dominating set $D'$ of size at most $|D|$ that does not contain any vertices from $B$.
\end{mylem}
\begin{proof}
Suppose we are given a minimum dominating set $D$ of $G$, where vertex $v \in B$ is present. In any dominating set, $s$ or $s'$ must be present. If $s'$ is present and $s$ is not, we replace $s'$ by vertex $s$, and still obtain a valid dominating set of the same size. As such, all vertices in $R$ are now dominated by $s$. Vertices $t^0_\ell$ and $t^1_\ell$ with $\ell \in [\log{\q}]$ are dominated by $s$.  Since $t^2_\ell$ only has neighbors $t^1_\ell$ and $t^0_\ell$, at least one of these three vertices is present in $D$ for every $\ell \in [\log \q]$, hereby every vertex in $T$ has a neighbor in $D$.

Since $B$ is an independent set in $G$, the vertex $v$ does not dominate other vertices in $B$. Since the polynomial equivalence relation ensures that there are no isolated vertices in $B$, vertex $v$ has at least one neighbor $u$ in $R$. We can safely replace $v$ by $u$ to obtain a valid dominating set that has the same size as $D$ and does not contain any vertices from $B$.
\end{proof}

\begin{mylem}
\label{lem:at_most_k_from_A}
Any dominating set of $G$ of size at most $ k+1+\log \q$ contains at least $1+\log{\q}$ vertices from $\{s,s'\} \cup \{t^0_\ell,t^1_\ell,t^2_\ell \mid \ell \in [\log \q]\}$ and thus contains at most $k$ vertices from $R$.
\end{mylem}
\begin{proof}
In a dominating set $D$ of $G$, at least $\log  \q$ vertices are needed from $T$, since $t^2_\ell$ only has neighbors $t^1_\ell$ and $t^0_\ell$, so one of these vertices must be in $D$ for each $\ell \in [\log \q]$. Furthermore at least one of the vertices $s'$ or $s$ must be present, therefore there are $1+\log{\q}$ vertices in the set that are not from $R$.
\end{proof}

\begin{mylem}
\label{lem:one_from_each_color_from_A}
Any dominating set of $G$ of size at most $k+1+\log \q$ uses exactly one vertex of each color from $R$.
\end{mylem}
\begin{proof}
Suppose a dominating set of $G$ of size at most $k+1+\log \q$ uses less than $k$ colors from $R$. If at most $k-2$ colors are used, there must be two colors $c_1$ and $c_2$ that are not present in the set. However, this implies that all $2K$ vertices in $W^{(c_1,c_2)}$ are not dominated by vertices in $R$ and must therefore be in the set. This contradicts the maximum size of the dominating set, since $K = k + 2 + \log \q$. So, we are left with the possibility of using $k-1$ colors. Consider some color $c_1$ that was not used. Look at another color $c_2$ that is used exactly once, such a color exists by Lemma \ref{lem:at_most_k_from_A}. Suppose the vertex of color $c_2$ in the dominating set was from set $R_i$ for some $i \in [{\q}]$.
Then for any $x \in \ID(R_i)$ we have that $w_x^{(c_1,c_2)}$ is not connected to any vertex in the dominating set and therefore must be in the dominating set itself. Since $\ID(R_i)$ contains $K$ numbers, there are $K$ vertices that are not dominated by $R$, which contradicts the maximum size of the dominating set.
\end{proof}


\begin{mylem}
\label{lem:all_in_one_square}
For any dominating set $D$ of $G$ of size at most $k+1+\log \q$, there exists $i \in [{\q}]$ such that all vertices in $D \cap R$ are contained in set $R_i$.
\end{mylem}
\begin{proof}
Suppose there exists two vertices $u,v \in D$ such that $u \in R_i$ and $v \in R_j$ for some $i \neq j$. By Lemma \ref{lem:one_from_each_color_from_A}, $u$ and $v$ have different colors. Suppose $u$ has color $c_u$ and $v$ has color $c_v$. Since $R_i \neq R_j$, there exists $x \in [2K]$ such that $x \in \ID(R_i)$ and $x \notin \ID(R_j)$. By Step \ref{ds:step:W} of the construction, this means that none of the neighbors of vertex $w_x^{(c_u, c_v)}$ are contained in the dominating set. However, this vertex is not in $D$ and therefore $D$ is not a dominating set of $G$, which is a contradiction.
\end{proof}

Using the previous Lemmas, we obtain:

\begin{mylem}
\label{lem:dominating_set_lr}
\begin{enumerate}
\item If there is an input $X_{i^*,j^*}$ that has a col-RBDS of size $k$, then $G'$ has a connected dominating set of size $k+1+\log{\q}$.
\item If $G'$ has a (not necessarily connected) dominating set of size $k+1+\log \q$, then some input $X_{i^*,j^*}$ has a col-RBDS of size $k$.
\end{enumerate}
\end{mylem}
\begin{proof}
$(1)$\quad Let $X_{i^*,j^*}$ have a colored RBDS $D$ of size at most $k$, then we can construct a dominating set $D'$ of $G$ in the following way. For any vertex $r^*_{p,q}$ in $D$, add vertex $r^i_{p,q}$ to $D'$.

Furthermore add the vertex $s$ to $D'$. Then add vertex $t_\ell^0$ to $D'$ if the $q$'th bit of $j^*$ is $1$, add vertex $t_\ell^1$ otherwise. Now $s'$ is dominated and all vertices in $R$ have neighbor $s$ in $D'$. All vertices in $B_{j^*}$ are covered by the vertices in the dominating set from $R_{i^*}$, since $D$ was a col-RBDS of $X_{i^*,j^*}$. All vertices in $B_j$ for $j \neq j^*$ have neighbor $t_\ell^0$ or $t_\ell^1$ in $D'$ for some $\ell \in [\log{\q}]$, since the bit representation of $j$ must differ from the one of $j^*$ at some position. It now follows from Step \ref{ds:step:T} of the construction that all vertices in $B_j$ are connected to a vertex in the dominating set.
It remains to verify that all vertices in $W$ have a neighbor in $D'$. Consider $w_x^{(c_1,c_2)}$ for $x \in [2K]$ and $c_1,c_2\in [k]$. If $x \in \ID(R_{i^*})$, then this vertex is connected to all vertices of color $c_1$ and exactly  one of them is contained in $D'$. If $x \notin \ID(R_{i^*})$, the vertex $w_x^{(c_1,c_2)}$ is connected to all vertices of color $c_2$ in $R_{i^*}$ and again one vertex of this color in $R_{i^*}$ is contained in $D'$. So $D'$ is a dominating set of $G$ and it is easy to verify that $|D'|=k + 1 + \log \q$. Furthermore, $D'$ is constructed in such a way that it is connected. We can show this by proving that every vertex in $D'$ is a neighbor of $s$, since we chose $s$ in $D$. Vertices in $D'\cap R$ and $D'\cap T$ are neighbors of $s$, by Steps \ref{ds:step:s} and \ref{ds:step:connected} of the construction of $G$. The vertex $s'$ and vertices from $W$ and $B$ are not contained in $D'$. Thus, $D'$ is a connected dominating set.

\noindent $(2)$ \quad Let $D'$ be a dominating set of $G$ of size at most $k + 1 + \log \q$. Using Lemma \ref{lem:no_from_B} we modify $D'$ such that it chooses no vertices from $B$, without increasing its size. By Lemma \ref{lem:one_from_each_color_from_A} and \ref{lem:all_in_one_square}, $D'$ contains exactly $k$ vertices from $R$, all from the same $R_{i^*}$ for some $i^*$ and all of different color. $D'$ has size at most $k + 1 + \log t$ of which $k$ are contained in $R$ and one in $\{s,s'\}$. Combined with the fact that for any $\ell \in [\log \q]$ vertex $t^2_\ell$ has $t^1_\ell$ and $t^0_\ell$ as its only two neighbors, it follows that exactly one of these three vertices is contained in $D'$ for all $\ell$. Therefore $D'$ contains at most one of the vertices $t_\ell^0$ or $t_\ell^1$ for every $\ell \in [\log \q]$.

We can now define $x_\ell \in \{0,1\}$ for $\ell \in [\log{\q}]$, such that $t_\ell^{x_\ell} \notin D'$ for all $\ell \in [\log{\q}]$.
Consider the index $j^* \in [t]$ given by the binary representation $[x_1\,x_2\,\ldots\,x_{\log{\q}}]_2$. 
  It follows from the bit representation of $j^*$ that the vertices in $B_{j^*}$ are not connected to any of the vertices in $D'\cap T$. Since vertices in $B_{j^*}$ are only adjacent to vertices in $R$ and vertices of $T$,   it follows that every vertex in $B_{j^*}$ has a neighbor in $R$ that is in  $D'$. This implies that every vertex in $B_{j^*}$ has a neighbor in $D'\cap R_{i^*}$. Since $G[R_{i^*} \cup B_{j^*}]$ is isomorphic to the graph of instance $X_{i^*,j^*}$, it follows that $X_{i^*,j^*}$ has a col-RBDS of size at most $k$, which are exactly the vertices in $D'\cap R_{i^*}$.
\end{proof}

Given $t$ instances, the graph $G$ constructed above has $n \cdot {\q} + m \cdot {\q} + 2 + 3 \cdot \log  \q + 2\binom{k}{2}\cdot 2K = \Oh(\sqrt{t} \max |X_{i,j}|^2)$ vertices. It is straightforward to construct $G$ in polynomial time. It follows from Lemma \ref{lem:dominating_set_lr}  that $G$ has a dominating set of size $k + 1 + \log \q$, if and only if one of the input instances has a col-RBDS of size $k$. Furthermore, $G$ has a connected dominating set of size $k+1+\log n$ if and only if one of the input instances has a col-RBDS of size $k$.
Therefore we have given a degree-$2$ cross-composition to \problem{(Connected) Dominating Set}. Using Theorem \ref{thm:cross_composition_LB} it follows that \problem{Dominating Set} and \problem{Connected Dominating Set} do not have a generalized kernel of size $\Oh(n^{2-\varepsilon})$ for any $\varepsilon > 0$, unless \containment.
\end{proof}

Just as the sparsification lower bounds for \textsc{Vertex Cover} that were presented by Dell and van Melkebeek~\cite{DellM14} had implications for the parameterization by the solution size~$k$, Theorem~\ref{thm:dominating_set} has implications for the kernelization complexity of \textsc{$k$-Nonblocker} and \textsc{$k$-Max Leaf}. Since the solution size~$k$ never exceeds the number of vertices in this problem, a kernel with~$\Oh(k^{2-\epsilon})$ edges would give a nontrivial sparsification, contradicting Theorem~\ref{thm:dominating_set}. Hence our results show that the existing linear-vertex kernels for \textsc{$k$-Nonblocker}~\cite[Corollary 4]{Wiedermann2006nonblocker} and \textsc{$k$-Max Leaf}~\cite{Estivill-CastroFLR05} cannot be improved to~$\Oh(k^{2-\varepsilon})$ edges unless \containment.

\section{\texorpdfstring{$d$-Hypergraph $2$-Colorability and $d$-NAE-SAT}{d-Hypergraph 2-Colorability and d-NAE-SAT}}
\label{sec:set_splitting}
The goal of this section is to give a nontrivial sparsification algorithm for \textsc{nae-sat} and prove a matching lower bound. For ease of presentation, we start by analyzing the closely related hypergraph $2$-colorability problem. Recall that a hypergraph consists of a vertex set~$V$ and a set~$E$ of \emph{hyperedges}; each hyperedge~$e \in E$ is a subset of~$V$. A $2$-coloring of a hypergraph is a function~$c \colon V \to \{1,2\}$; such a coloring is \emph{proper} if there is no hyperedge whose vertices all obtain the same color. We will use \textsc{$d$-Hypergraph $2$-Colorability} to refer to the setting where hyperedges have size at most~$d$. The corresponding decision problem asks, given a hypergraph, whether it is $2$-colorable.

\begin{mythm}\label{thm:set_splitting_kernel}
\problem{$d$-Hypergraph $2$-Colorability} parameterized by the number of vertices~$n$ has a kernel with $2 \cdot n^{d-1}$ hyperedges that can be encoded in $\Oh(n^{d-1} \cdot d \cdot \log n)$ bits.
\end{mythm}
\begin{proof}
Suppose we are given a hypergraph with vertex set~$V$ and hyperedges~$E$, where each hyperedge contains at most $d$ vertices. We show how to reduce the number of hyperedges without changing the $2$-colorability status. Let $E_r \subseteq E$ denote the set of edges in $E$ that contain exactly $r$ vertices. For each $E_r$ we construct a set $E'_r \subseteq E_r$ of \emph{representative hyperedges}. Enumerate the edges in $E_r$ as $e^r_1,\ldots,e^r_k$. We construct a $(0,1)$-matrix $M_r$ with $N:=\binom{n}{r-1}$ rows and $k$ columns. Consider all possible subsets $A_1,\ldots,A_N$ of size $r-1$ of the set of vertices $V$. Define the elements $m_{i,j}$ for $i\in N$ and $j \in k$ of $M_r$ as follows.
\[ m_{i,j} := \left\{ \begin{array}{ll}
         1 & \mbox{if $A_i \subseteq e^r_j$};\\
        0 &  \mbox{otherwise.} \end{array} \right. \]
Using Gaussian elimination, compute a basis $B$ of the columns of this matrix, which is a subset of the columns that span the column space of $M_r$. Let $E'_r$ contain edge $e^r_i$ if the $i$'th column of $M_r$ is contained in $B$, and define $E' := \bigcup_{r \in [d]} E'_r$, which forms the kernel. Using a lemma due to Lov\'{a}sz~\cite{lovasz}, we can prove that~$E'$ preserves the $2$-colorability status.

\begin{mylem}[\cite{lovasz}]
\label{lem:d_set_splitting_matrix_lemma}
Let $H$ be an $r$-uniform hypergraph with edges $E_1,\ldots,E_m$. Let $\alpha_1,\ldots,\alpha_m$ be real numbers such that for every $(r-1)$-element subset $A$ of $V(H)$,
\[\sum_{E_i \supset A} \alpha_i = 0.\]
Then for every partition $\{V_1,V_2\}$ of $V(H)$ the following holds:
\[\sum_{E_i\subseteq V_1} \alpha_i = (-1)^r \sum_{E_i \subseteq V_2} \alpha_i.\]
\end{mylem}

Now we can prove the correctness of the presented kernel. 

\begin{mylem}
\label{lem:dss:kernel_valid}
$(V,E)$ has a proper $2$-coloring $\Leftrightarrow$ $(V,E')$ has a proper $2$-coloring.
\end{mylem}
\begin{proof}
($\Rightarrow$) Clearly, if $(V,E)$ has a proper $2$-coloring, then the same coloring is proper for the subhypergraph~$(V,E')$ since~$E' \subseteq E$.

($\Leftarrow$) Now suppose $(V,E')$ has a proper $2$-coloring. We will show that for each $r \in [d]$, no edge of~$E_r$ is monochromatic under this coloring. All hyperedges contained in $E'_r$ are $2$-colored by definition. Suppose there exists $r \in [d]$, such that $E_r$ contains a monochromatic hyperedge. Let~$E_r = e^r_1,\ldots,e^r_k$ and let $e_{i^*}$ be a hyperedge in~$E_r$ whose vertices all receive the same color.

By reordering the matrix~$M_r$, we may assume that the basis $B$ of $M_r$ contains the first $\ell$ columns, thus $i^* > \ell$. Let $\vect{m}_i$ denote the $i$'th column of $M_r$.  Since $\vect{m}_{i^*}$ is not contained in the basis, there exist coefficients $\alpha_1,\ldots,\alpha_\ell$ such that
\[\sum_{i=1}^\ell \alpha_i\cdot \vect{m}_i = \vect{m}_{i^*}. \]
For $i \in [k]$, define:
\[ \beta_i := \left\{ \begin{array}{lll}
         \alpha_i & \mbox{if $i \leq \ell$};\\
        -1 &  \mbox{if $i=i^*$};\\
				 0 &  \mbox{otherwise.} \end{array} \right. \]
From this definition of $\beta$ it follows that
\begin{align*}
&\sum_{i=1}^{k} \beta_i\cdot\vect{m}_i = \sum_{i=1}^\ell \alpha_i\cdot \vect{m}_i - \vect{m}_{i^*}= \vect{0}. \\
\intertext{Let $A_j$ be any size $(r-1)$-subset of $V$. Since $m_{i,j} = 1$ exactly when $e_i \supseteq A_j$, and $0$ otherwise, we have:}
&\sum_{e_i \supset A_j} \beta_i = \sum_{i=1}^{k} \beta_i m_{i,j} = 0.
\end{align*}
By Lemma \ref{lem:d_set_splitting_matrix_lemma} we obtain that for any partitioning $V_1 \cup V_2$ of the vertices in $V$,
\begin{align}
\label{dss:formula}
\sum_{e_i\subseteq V_1} \beta_i = (-1)^r \sum_{e_i \subseteq V_2} \beta_i.
\end{align}
Consider however the partitioning $(V_1, V_2)$ given by the $2$-coloring of the vertices. Then every edge $e_i \in E'_r$ contains at least one vertex of each color and is thereby not fully contained in $V_1$ or $V_2$. As such, these edges contribute $0$ to both sides of the equation. The edge $e_{i^*}$ is the only remaining edge with a non-zero coefficient and by assumption, it is contained entirely within one color class. Without loss of generality, let $e_{i^*} \subseteq V_1$. But then $\sum_{e_i\subseteq V_1} \beta_i = -1$ while $(-1)^r \sum_{e_i \subseteq V_2} \beta_i = 0$, which contradicts (\ref{dss:formula}).
\end{proof}

To bound the size of the kernel, consider the matrix~$M_r$ for $r \in[d]$. Its rank is bounded by the minimum of its number of rows and columns, which is at most $\binom{n}{r-1} \leq n^{r-1}$. As such, we get $|E'_r| \leq \mathrm{rank}(M_r) \leq n^{r-1}$. Note that $d \leq n$, such that $|E'| \leq \sum_{r=1}^d n^{r-1} = n^{d-1} +  \sum_{r=1}^{d-1} n^{r-1} \leq 2\cdot n^{d-1}$.
So $E'$ contains at most $2n^{d-1}$ hyperedges. Since a hyperedge consists of at most $d$ vertices, the kernel can be encoded in $\Oh(n^{d-1}\cdot d \cdot \log{n})$ bits.
\end{proof}

By a folklore reduction, Theorem~\ref{thm:set_splitting_kernel} gives a sparsification for \textsc{nae-sat}. Consider an instance of \textsc{$d$-nae-sat}, which is a conjunction of clauses of size at most~$d$ over variables~$x_1, \ldots, x_n$. The formula gives rise to a hypergraph on vertex set~$\{x_i, \neg x_i \mid i \in [n]\}$ containing one hyperedge per clause, whose vertices correspond to the literals in the clause. When additionally adding~$n$ hyperedges~$\{x_i, \neg x_i\}$ for~$i \in [n]$, it is easy to see that the resulting hypergraph is $2$-colorable if and only if there is a NAE-satisfying assignment to the formula. The maximum size of a hyperedge matches the maximum size of a clause and the number of created vertices is twice the number of variables. We can therefore sparsify an $n$-variable instance of \textsc{$d$-nae-sat} in the following way: reduce it to a $d$-hypergraph with~$n' := 2n$ vertices and apply the kernelization algorithm of Theorem~\ref{thm:set_splitting_kernel}. It is easy to verify that restricting the formula to the representative hyperedges in the kernel gives an equisatisfiable formula containing~$2 \cdot (n')^{d-1} \in \Oh(2^{d-1} n^{d-1})$ clauses, giving a sparsification for \textsc{nae-sat}. As mentioned in the introduction, the existence of a linear-parameter transformation~\cite[p.81]{JansenK13} from \textsc{$d$-cnf-sat} to \textsc{($d+1$)-nae-sat} also implies a sparsification \emph{lower bound} for \textsc{$d$-nae-sat}, using the results of Dell and van Melkebeek~\cite{DellM14}. Hence we obtain the following theorem.

\begin{theorem}
For every fixed~$d \geq 4$, the \textsc{$d$-nae-sat} problem parameterized by the number of variables~$n$ has a kernel with~$\Oh(n^{d-1})$ clauses that can be encoded in~$\Oh(n^{d-1} \cdot \log n)$ bits, but admits no generalized kernel of size~$\Oh(n^{d-1-\varepsilon})$ for~$\varepsilon > 0$ unless \containment.
\end{theorem} 

\section{Conclusion}
\label{sec:conclusion}
We have added several classic graph problems to a growing list of problems for which non-trivial polynomial-time sparsification is provably impossible under the assumption that \ncontainment. Our results for \textsc{(Connected) Dominating Set} proved that the linear-vertex kernels with~$\Theta(k^2)$ edges for \textsc{$k$-Nonblocker} and \textsc{$k$-Max Leaf Spanning Tree} cannot be improved to~$\Oh(k^{2-\varepsilon})$ edges unless \containment. 

The graph problems for which we proved sparsification lower bounds can be defined in terms of vertices: the \textsc{$4$-Coloring} problem asks for a partition of the vertex set into four independent sets, \textsc{Dominating Set} asks for a dominating subset of vertices, and \textsc{Hamiltonian Cycle} asks for a permutation of the vertices that forms a cycle. In contrast, not much is known concerning sparsification lower bounds for problems whose solution is an edge subset of possibly quadratic size. For example, no sparsification lower bounds are known for well-studied problems such as \textsc{Max Cut}, \textsc{Cluster Editing}, or \textsc{Feedback Arc Set in Tournaments}. Difficulties arise when attempting to mimic our lower bound constructions for such edge-based problems. Our constructions all embed~$t$ instances into a~$2 \times \sqrt{t}$ table, using each combination of a cell in the top row and bottom row to embed one input. For problems defined in terms of edge subsets, it becomes difficult to ``turn off'' the contribution of edges that are incident on vertices that do not belong to the two cells that correspond to a \yes-instance among the inputs to the \textsc{or}-construction. This could be interpreted as evidence that edge-based problems such as \textsc{Max Cut} might admit non-trivial polynomial sparsification. We have not been able to answer this question in either direction, and leave it as an open problem. For completeness, we point out that Karp's reduction~\cite{Karp72} from  \textsc{Vertex Cover} to \textsc{Feedback Arc Set} (which only doubles the number of vertices) implies, using existing bounds for \textsc{Vertex Cover}~\cite{DellM14}, that \textsc{Feedback Arc Set} does not have a compression of size~$\Oh(n^{2-\varepsilon})$ unless \containment.

Another problem whose compression remains elusive is \textsc{$3$-Coloring}. In several settings (cf.~\cite{FellowsJ14}), the optimal kernel size matches the size of minimal obstructions in a problem-specific partial order. This is the case for \textsc{$d$-nae-sat}, whose kernel with~$\Oh(n^{d-1})$ clauses matches the fact that critically 3-chromatic $d$-uniform hypergraphs have at most~$\Oh(n^{d-1})$ hyperedges. Following this line of reasoning, it is tempting to conjecture that \textsc{$3$-Coloring} does not admit subquadratic compressions: there are critically $4$-chromatic graphs with~$\Theta(n^2)$ edges~\cite{Toft70}.

The kernel we have given for \textsc{$d$-nae-sat} is one of the first examples of non-trivial polynomial-time sparsification for general structures that are not planar or similarly guaranteed to be sparse. Obtaining non-trivial sparsification algorithms for other problems is an interesting challenge for future work. Are there natural problems defined on general graphs that admit subquadratic sparsification?

\bibliographystyle{plain}
\bibliography{lipics_report}

\newpage






\end{document}